\documentclass[11pt]{article} 
\usepackage{fullpage}
\usepackage{amsmath}
\usepackage{amssymb}
\usepackage{booktabs}
\usepackage{multirow}
\usepackage{footmisc}
\usepackage{xspace}
\usepackage[section]{placeins}
\usepackage[algosection]{algorithm2e}
\usepackage{amsfonts}
\usepackage{enumerate}

\usepackage{amsthm}

\newtheorem{theorem}{Theorem}[section]
\newtheorem{corollary}[theorem]{Corollary}
\newtheorem{lemma}[theorem]{Lemma}
\newtheorem{proposition}[theorem]{Proposition}

\newtheorem*{claim}{Claim}

\newenvironment{theorem_cite}[1]
{\begin{theorem}  {\rm (#1)}}
{\end{theorem}}

\theoremstyle{definition}

\newtheorem*{definition}{Definition}
\newtheorem*{defn}{Definition}

\newtheorem{example}[theorem]{Example}

\newtheorem*{example*}{Example}
\newtheorem*{examples*}{Examples}

\theoremstyle{remark}

\numberwithin{equation}{section}
\numberwithin{figure}{section}

\numberwithin{theorem}{section}
\numberwithin{table}{section}
\numberwithin{figure}{section}
\numberwithin{equation}{section}

\title {\bf Nonuniform Reductions and NP-Completeness}
\author{
  John M. Hitchcock
  and
  Hadi Shafei\\
  Department of Computer Science\\
  University of Wyoming
}
\date{}

\usepackage{amsfonts}
\usepackage{amsmath} 
\usepackage{amssymb}
\usepackage{enumerate}

\newcommand{\ignore}[1]{}

\newcommand{\SAT}{{\rm SAT}}

\newcommand{\PROOF}{\begin{proof}}%

\newcommand{\QED}{\end{proof}}

\newcommand{\N}{\mathbb{N}}

\newcommand{\p}{{\mathrm{p}}}
\newcommand{\ptwo}{{\p_{\thinspace\negthinspace_2}}}

\newcommand{\mup}{\mu_{\mathrm{p}}}

\newcommand{\muptwo}{\mu_\ptwo}

\newcommand{\co}[1]{\mathrm{co}#1}

\newcommand{\NP}{{\ensuremath{\mathrm{NP}}}} 
 
\newcommand{\coNP}{\co{\NP}}

\newcommand{\SUBEXP}{\mathrm{SUBEXP}}

\newcommand{\BPP}{{\rm BPP}}

\newcommand{\E}{{\rm E}}

\newcommand{\PSPACE}{{\rm PSPACE}}

\renewcommand{\P}{\ensuremath{{\mathrm P}}}

\newcommand{\EXP}{{\rm EXP}}
\newcommand{\NEXP}{{\rm NEXP}}

\newcommand{\SigmaEXPkPiEXPk}{\SigmaEXPkPiEXPk}

\newcommand{\PH}{{\rm PH}}

\newcommand{\poly}{{\rm poly}}

\newcommand{\Ppoly}{{\P/\poly}}

\newcommand{\PF}{\mathrm{PF}}

\newcommand{\calH}{{\cal H}}

\newcommand{\calR}{{\cal R}}

\newcommand{\KS}{{\mathrm{KS}} }

\newcommand{\Kt}{{\mathrm{Kt}}}

\newcommand{\io}{{\mathrm{i.o.}}}

\newcommand{\nth}{n^{\mathrm{th}}}

\renewcommand{\tt}{\mathrm{tt}}
\newcommand{\T}{\mathrm{T}}
\newcommand{\leqpm}{\leqp_\mathrm{m}}

\newcommand{\leqp}{\leq^\p}
\newcommand{\leqpT}{\leqp_\T}

\newcommand{\leqptt}{\leqp_\mathrm{tt}}

\newcommand{\m}{\mathrm{m}}

\renewcommand{\leqp}{\leq^\P}
\newcommand{\leqpadv}[2]{\leq^{\P/#1}_{#2}}
\newcommand{\leqpmadv}[1]{\leqpadv{#1}{\m}}
\newcommand{\leqpTadv}[1]{\leqpadv{#1}{\T}}
\newcommand{\leqpttadv}[1]{\leqpadv{#1}{\tt}}

\begin{document}
\maketitle

\begin{abstract}

Nonuniformity is a central concept in computational complexity with
powerful connections to circuit complexity and randomness. Nonuniform
reductions have been used to study the isomorphism conjecture for NP
and completeness for larger complexity classes. We study the power of
nonuniform reductions for NP-completeness, obtaining both separations
and upper bounds for nonuniform completeness vs uniform complessness
in NP.

Under various hypotheses, we obtain the following separations:
\begin{enumerate}
\item There is a set complete for NP under nonuniform many-one
  reductions, but not under uniform many-one reductions. This is true
  even with a single bit of nonuniform advice.

\item There is a set complete for NP under nonuniform many-one
  reductions with polynomial-size advice, but not under uniform Turing
  reductions. That is, polynomial nonuniformity is stronger than a
  polynomial number of queries.

\item For any fixed polynomial $p(n)$, there is a set complete for NP under
  uniform 2-truth-table reductions, but not under nonuniform many-one
  reductions that use $p(n)$ advice. That is, giving a uniform
  reduction a second query makes it more powerful than a nonuniform
  reduction with fixed polynomial advice.

\item There is a set complete for NP under nonuniform many-one
  reductions with polynomial advice, but not under nonuniform many-one
  reductions with logarithmic advice. This hierarchy theorem also
  holds for other reducibilities, such as truth-table and Turing.

\end{enumerate}

We also consider uniform upper bounds on nonuniform completeness.
  Hirahara (2015) showed that unconditionally every set that is
  complete for NP under nonuniform truth-table reductions that use
  logarithmic advice is also uniformly Turing-complete.
  We show that under a derandomization hypothesis, the same statement
  for truth-table reductions and truth-table completeness also holds.

\end{abstract}

\section{Introduction}

Nonuniformity is a powerful concept in computational
complexity. 
In a nonuniform computation a
  different algorithm or circuit may be used for each input size
  \cite{Shannon49}, as opposed to a uniform computation in which a
  single algorithm must be used for all inputs. Alternatively,
  nonuniform advice may be provided to a uniform algorithm --
  information that may not be computable by the algorithm but is
  computationally useful \cite{KarLip82}. 
For example, nonuniformity can be used as a substitute for randomness
\cite{Adle78}: every randomized algorithm can be replaced by a
nonuniform one ($\BPP \subseteq \Ppoly$). It is unknown whether the
same is true for $\NP$, but the Karp-Lipton Theorem \cite{KarLip82}
states that if the polynomial-time hierarchy does not collapse, then
$\NP$-complete problems have superpolynomial nonuniform complexity
($\PH$ is infinite implies $\NP \not\subseteq \Ppoly$). Hardness
versus randomness tradeoffs show that such nonuniform complexity lower
bounds imply derandomization (for example, $\EXP \not \subseteq
\Ppoly$ implies $\BPP \subseteq \io\SUBEXP$ \cite{BFNW93}).

Nonuniform computation can also be used to give reductions between
decision problems, when uniform reductions are lacking. The
Berman-Hartmanis Isomorphism Conjecture \cite{BerHar77} for $\NP$
asserts that all $\NP$-complete sets are isomorphic under
polynomial-time reductions. Progress towards relaxations of the
Isomorphism Conjecture with nonuniform reductions has been made
\cite{Agr02,AgrWat09,Hitchcock:CRNPCS} under various hypotheses.

Allender et al. \cite{Allender:power} used nonuniform reductions to
investigate the complexity of sets of Kolmogorov-random strings. They
showed that the sets $R_{\KS}$ and $R_{\Kt}$ are complete for
$\PSPACE$ and $\EXP$, respectively, under 
$\Ppoly$-truth-table reductions.  $R_{\Kt}$ is not complete under
polynomial-time truth-table reductions -- in fact, the full
polynomial-size advice is required \cite{Ronneburger:phdthesis}.

The Minimum Circuit Size Problem (MCSP) \cite{KabanetsCai00} is an
intriguing $\NP$ problem. It is not known to be
$\NP$-complete. Proving it is $\NP$-complete would imply consequences
we don't yet know how to prove, yet there is really no strong evidence
that it isn't $\NP$-complete. Recently Allender
\cite{Allender:complexityofcomplexity} has asked if the Minimum
Circuit Size Problem \cite{KabanetsCai00} is $\NP$-complete under
$\Ppoly$-Turing reductions.

Buhrman et
al. \cite{BuhrmanHescottHomerTorenvliet:NonuniformReductions} began a
systematic study of nonuniform completeness.
They proved, under a strong hypothesis on $\NP$, that every 1-$\tt$-complete set for $\NP$ is many-one complete with 1 bit of advice. This result has been known for larger classes like 
$\EXP$ and $\NEXP$ without using any advice. They also proved a separation between uniform and nonuniform reductions in $\EXP$ by showing that there exists a language that is complete in  $\EXP$ under many-one reductions that use one bit of advice, but is not 2-tt-complete \cite{BuhrmanHescottHomerTorenvliet:NonuniformReductions}. They also proved that a nonuniform reduction can be turned into a uniform one by increasing the number of queries. 

While Buhrman et
al. \cite{BuhrmanHescottHomerTorenvliet:NonuniformReductions} have
some results about nonuniform reductions in $\NP$, most of their
results are focused on larger complexity classes like $\EXP$.
Inspired by their results on $\EXP$, we work toward a similarly solid 
understanding of NP-completeness under nonuniform reductions. We give
both separation and upper bound results for a variety of nonuniform
and uniform completeness notions. We consider the standard
polynomial-time reducibilities including many-one ($\leqpm$),
truth-table ($\leqptt$), and Turing ($\leqpT$). We will consider
nonuniform reductions such as $\leqpmadv{h(n)}$ where the algorithm
computing the reduction is allowed $h(n)$ bits of advice for inputs of
size $n$.

\paragraph*{Separating Nonuniform Completeness from Uniform
  Completeness.}  

We show in Section \ref{sec:nonuniform-uniform} that nonuniform
reductions can be strictly more powerful than uniform reductions for
NP-completeness. This is necessarily done under a hypothesis, for if
$\P = \NP$, all completeness notions for $\NP$ trivially collapse. We
use the Measure Hypothesis and the NP-Machine Hypothesis -- two
hypotheses on $\NP$ that have been used in previous work to separate
NP-completeness notions \cite{Lutz:CVKL,Pavan03,Hitchcock:HHDCC}. The
Measure Hypothesis asserts that $\NP$ does not have $\p$-measure 0
\cite{Lutz:AEHNC,Lutz:QSET}, or equivalently, that $\NP$ contains a
$\p$-random set \cite{AmTeZh97,AmbMay97}. The NP-Machine Hypothesis
\cite{Hitchcock:HHDCC} has many equivalent formulations and implies
that there is an $\NP$ search problem that requires exponential time
to solve almost everywhere.

We show under the {Measure Hypothesis} that there is a
$\leqpmadv{1}$-complete set for $\NP$ that is not
$\leqpm$-complete. In other words, nonuniform many-one reductions are
stronger than many-one reductions for NP-completeness, and this holds with even a single nonuniform advice bit.

We also show that if the nonuniform reductions are allowed more
advice, we have a separation even from Turing reductions. Under the
{NP-Machine Hypothesis}, there is a $\leqpmadv{\poly}$-complete set
that is not $\leqpT$-complete. That is, polynomial-size advice makes a
many-one reduction stronger for $\NP$-completeness than a reduction
that makes a polynomial number of adaptive queries.

\paragraph*{Separating Uniform Completeness from Nonuniform
  Completeness} 
Next, in Section \ref{sec:uniform-nonuniform}, we give
evidence that uniform reductions may be strictly stronger than
nonuniform reductions for $\NP$-completeness.

We show under a hypothesis on $\NP \cap \coNP$ that adding just one
more query makes a reduction more powerful than a nonuniform one for
completeness: if $\mup(\NP \cap \coNP) \neq 0$, then for any $c \geq
1$, there is a $\leqp_{2-\tt}$-complete set that is not
$\leqpmadv{n^c}$-complete.  This is an interesting contrast to our
separation of $\leqpmadv{\poly}$-completeness from
$\leqpT$-completeness (which includes
$\leqp_{2-\tt}$-completeness). Limiting the advice on the many-one
reduction to a fixed polynomial flips the separation the other way --
and in fact, only two queries are needed. The $\mup(\NP \cap \coNP)
\neq 0$ hypothesis is admittedly strong.  However, we note that strong
hypotheses on $\NP \cap \coNP$ have been used in some prior
investigations
\cite{PavSel04,Hitchcock:ANPC,BuhrmanHescottHomerTorenvliet:NonuniformReductions}.

\paragraph*{Uniform Completeness Upper Bounds for Nonuniform Completeness}

Despite the above separations, it is possible to replace a limited
amount of nonuniformity by a uniform reduction for NP-completeness. Up to
logarithmic advice may be made uniform at the expense of a polynomial
number of queries:

\begin{enumerate}
\item A result of Hirahara \cite{Hirahara15} implies every
  $\leqpTadv{\log}$-complete set for $\NP$ is also $\leqpT$-complete.

\item Under a derandomization hypothesis ($\E$ has a problem with high
  $\NP$-oracle circuit complexity), we show that every
  $\leqpttadv{\log}$-complete set for $\NP$ is also
  $\leqptt$-complete.
The Valiant-Vazirani lemma \cite{ValVaz86} gives a
  randomized algorithm to reduce the satisfiability problem to the
  unique satisifability problem. Being able to derandomize this
  algorithm \cite{KlivMe02} yields a nonadaptive reduction.
\end{enumerate}
These upper bound results are presented in Section \ref{sec:upperbound}.

\paragraph*{Hierarchy Theorems for Nonuniform Completeness}
In Section \ref{sec:hierarchy}, we give hierarchy theorems for
nonuniform NP-completeness.
We separate
polynomial advice from logarithmic advice: if the
$\NP$-machine hypothesis is true, then there is a
$\leqpmadv{\poly}$-complete set that is not
$\leqpmadv{\log}$-complete. This also holds for other reducibilities
such as truth-table and Turing.

\section{Preliminaries}

All languages in this paper are subsets of $\{0,1\}^*$. We use the standard enumeration of binary strings, i.e. $s_0 = \lambda, s_1 = 0, s_2 = 1,s_3 = 00,...$ as an order on binary strings. For any language $A \subseteq
\{0,1\}^*$ the characteristic sequence of $A$ is defined as $\chi_A = A[s_0]A[s_1]A[s_2]...$ where 
$A[x] = 1$ or $0$ depending on whether the string $x$ belongs to $A$ or not respectively. We identify every language with its characteristic sequence.  For any binary sequence $X$ and any string $x \in \{0,1\}^*$, $X\upharpoonright x$ is the initial segment of $X$ for all strings before $x$. 

We use the standard definitions of complexity classes and well-known
reductions that can be found in \cite{BaDiGa95,Papa94}. For any two
languages $A$ and $B$ and a function $l : \N \to \N$, we say $A$ is
{\em{nonuniform polynomial-time reducible to $B$ with advice $l(n)$}},
and we write $A \leqpmadv{l(n)} B$, if there exists $f \in \PF$ and $h
: \N \to \{0,1\}^*$ with $|h(n)| \leq l(n)$ for all $n$ such that
$(\forall  x) \; x \in A \; \leftrightarrow f \big(x , h(|x|)\big) \in B. $
The string $h(|x|)$ is called the {\em{advice}}, and it only
depends on the length of the input. For a class $\calH$ of functions
mapping $\N \to \{0,1\}^*$, we say $A \leqpmadv{\calH} B$ if $A
\leqpmadv{l} B$ for some $l \in \calH$. The class $\poly$ denotes all
advice functions with length bounded by a polynomial, and $\log$ is
all advice functions with length $O(\log n)$.
We also use $\leqpmadv{1}$ for a nonuniform reduction when $|h(|x|)| = 1$.
Nonuniform reductions can similarly be defined
with respect to other kinds of reductions like Turing, truth-table,
etc.

In most of our proofs we use resource-bounded measure
\cite{Lutz:AEHNC} to state our hypotheses. In the following we provide
a brief description of this concept. For more details, see
\cite{Lutz:AEHNC,Lutz:QSET,AmbMay97}. A {\em{martingale}} is a
function $d : \{0,1\}^* \rightarrow [0, \infty)$ where $d(\lambda) >
  0$ and $\forall x \in \{0 , 1 \}^*$, $2d(x) = d(x0) + d(x1)$.
We say a martingale {\em{succeeds}} on a set $A \subseteq \{0,1\}^*$ if
$\limsup_{n \rightarrow \infty} \; d(A \upharpoonright n) = \infty,$
where $A \upharpoonright n$ is the length $n$ prefix of $A$’s characteristic sequence.
One can think of the martingale $d$ as a strategy for betting on the
consecutive bits of the characteristic sequence of $A$. The martingale
is allowed to use the first $n - 1$ bits of $A$ when betting on the
$\nth$ bit. Betting starts with the initial capital $d(\lambda) > 0$,
and $d(A \upharpoonright n - 1)$ denotes the capital after betting on
the first $(n-1)$ bits. At this stage the martingale bets some amount
$a$ where $0 \leq a \leq d(A \upharpoonright n - 1)$ that the next bit
is $0$ and the rest of the capital, i.e. $d(A \upharpoonright n - 1) -
a$, that the next bit is $1$. If the $\nth$ bit is $0$, then $d(A
\upharpoonright n) = 2a$. Otherwise, $d(A \upharpoonright n) = 2(d(A
\upharpoonright n - 1) - a)$.  For any time bound $t(n)$, we say a
language $L$ is $t(n)$-random if no $O(t(n))$-computable martingale
succeeds on $L$. A language is $\p$-random if it is $n^c$-random for
every $c$. A language is $\ptwo$-random if it is $2^{{\log n}^c}$-random for some $c$. A class of languages $C$ has $\p$-measure 0,
written $\mup(C) = 0$, if there is a $c$ such that no
language in $C$ is $n^c$-random. Similarly, $C$ has $\ptwo$-measure
0, written $\muptwo(C) = 0$, if there is a $c$ such that no language
in $C$ is $2^{\log^c n}$-random. If $C$ is closed under
$\leqpm$-reductions, then $\mup(C) = 0$ if and only if $\muptwo(C) =
0$ \cite{Lutz:WCE}.

We will use the {\em Measure Hypothesis} that $\mup(\NP) \neq 0$ and
the {\em \NP-Machine Hypothesis} \cite{Hitchcock:HHDCC}: there is an
$\NP$ machine $M$ and an $\epsilon > 0$ such that $M$ accepts $0^*$
and no $2^{n^\epsilon}$-time-bounded Turing machine computes
infinitely many accepting computations of $M$. The Measure Hypothesis
implies the $\NP$-Machine Hypothesis \cite{Hitchcock:HHDCC}.

\section{Separating Nonuniform Completeness from Uniform Completeness}\label{sec:nonuniform-uniform}
Our first theorem separates nonuniform many-one completeness with one
bit of advice from uniform many-one completeness for $\NP$, under the
measure hypothesis. Buhrman et
al. \cite{BuhrmanHescottHomerTorenvliet:NonuniformReductions} proved
the same result for $\EXP$ unconditionally.

\begin{theorem}
If $\mup(\NP) \neq 0$ then there exists a set $D \in \NP$ that is $\NP$-complete with respect to $\leqpmadv{1}$-reductions but is not $\leqpm$-complete.
\end{theorem}
\begin{proof}
Let $R \in \NP$ be $\p$-random. We use $R$ and $\SAT$ to construct the following set:
\[ D = 
  \begin{tabular}{ccc}
  $\langle \phi, 0 \rangle$ : $\phi \in \SAT \;\vee \; 0^{|\phi|} \in
    R \rangle $ \\
  $\bigcup \langle \phi, 1 \rangle$ : $\phi \in \SAT \;\wedge \;
    0^{|\phi|} \in R \rangle$
  \end{tabular}
\]
It follows from closure properties of $\NP$ that $D \in \NP$. It is also easy to see that 
$\SAT \leqpmadv{1} D$ via $\phi \rightarrow \langle \phi , R[0^{|\phi|}] \rangle$. Note that $R[0^{|\phi|}]$ is one bit of advice, and it is $1$ or $0$ depending on whether or not $0^{|\phi|} \in R$. We will prove that $D$ is not $\leqpm$-complete for $\NP$. To get a contradiction, assume that $D$ is $\leqpm$-complete. Therefore $\SAT \leqpm D$ via some polynomial-time computable function $f$. Then
$(\forall \phi) \; \phi \in \SAT \leftrightarrow f(\phi) \in D.$
Based on the value of $\SAT[\phi]$ and the second component of $f(\phi)$ we consider four cases:
\begin{enumerate}

\item $\phi \in \SAT \; \wedge f(\phi) = \langle \psi , 0 \rangle$, for some formula $\psi$.
\item $\phi \notin \SAT \; \wedge f(\phi) = \langle \psi , 0 \rangle$, for some formula $\psi$.
\item $\phi \in \SAT \; \wedge f(\phi) = \langle \psi , 1 \rangle$, for some formula $\psi$.
\item $\phi \notin \SAT \; \wedge f(\phi) = \langle \psi , 1 \rangle$, for some formula $\psi$.

\end{enumerate}

In the second case above we have $\SAT[\phi] = \SAT[\psi] \vee
R[0^{|\psi|}]$ and $\phi \notin \SAT$. Therefore $\SAT[\psi] \vee
R[0^{|\psi|}] = 0$ which implies $R[0^{|\psi|}] = 0$. Consider the
situation where the second case happens and $|\psi| \geq
|\phi|/2$. The following  argument shows that  $R[0^{|\psi|}]$ is
computable in $2^{5|\psi|}$ time in this situation. We apply  $f$ to every string of length at most 
$2|\psi|$, looking for a formula $\phi$ of length at most $2|\psi|$ such that $f(\phi) = \langle \psi , 0 \rangle$ and $\phi \notin \SAT$. We are applying $f$ which is computable in polynomial time to at most $2^{2|\psi|+1}$ strings. This can be done in $2^{3|\psi|}$ steps. Checking if $\phi \notin \SAT$ can be done in at most $2^{2|\psi|}$ steps for each 
$\phi$. Therefore the whole algorithm takes at most $2^{5|\psi|}$ steps to terminate. If this case happens for infinitely many $\psi$'s we will have a polynomial-time martingale that succeeds on $R$ which contradicts the $\p$-randomness of $R$. As a result, there cannot be infinitely many $\phi$'s that  $\phi \notin \SAT $, $f(\phi) = \langle \psi , 0 \rangle$, and $|\psi| \geq |\phi|/2$. This is because if there are infinitely many such $\phi$'s, then there must be infinitely many $n$'s such that for each $n$ there exists a $\phi$ satisfying the above properties. Since we assumed $|\psi| \geq |\phi|/2$ it follows that there must be infinitely many such $\psi$'s, but we proved that this cannot happen. 

An analagous argument for the third case 
there cannot be infinitely many $\phi$'s that  $\phi \notin \SAT $,
$f(\phi) = \langle \psi , 0 \rangle$, and $|\psi| \geq |\phi|/2$. 
Therefore we have:
\begin{enumerate}
\item For almost every $\phi$, if $\phi \notin \SAT \; \wedge f(\phi) = \langle \psi , 0 \rangle$, then $|\psi| < |\phi|/2$. 
\item For almost every $\phi$, if $\phi \in \SAT \; \wedge f(\phi) = \langle \psi , 1 \rangle$, then $|\psi| < |\phi|/2$.
\end{enumerate}

It follows from these two facts that for almost every $\phi$, if $|\psi| \geq |\phi|/2$, then $\SAT[\phi]$ can be computed in polynomial time: \\
\begin{enumerate}
\item If $f(\phi) = \langle \psi , 0 \rangle$ and $|\psi| \geq |\phi|/2$, then $\phi \in \SAT$. 
\item If $f(\phi) = \langle \psi , 1 \rangle$ and $|\psi| \geq |\phi|/2$, then $\phi \notin \SAT$. 
\end{enumerate}

Note that the only computation required in the algorithm above is computing $f$ on $\phi$ which can be done in polynomial time.
To summarize, for every formula $\phi$ it is either the case that when we apply $f$ to $\phi$ the new formula $\psi$ satisfies $|\psi| < |\phi|/2$ or $\SAT[\phi]$ is computable in polynomial time. In the following we use this fact and the many-one reduction from $\SAT$ to $D$ to introduce a $(\log n)$-$\tt$-reduction from $\SAT$ to $R$.

The many-one reduction from $\SAT$ to $D$ implies that
$(\forall \phi) \; \phi \in \SAT \leftrightarrow f(\phi) \in D.$
  In other words:
\begin{equation}\label{second-reduction}
(\forall \phi) \; f(\phi) =   \langle \psi_1 , i \rangle 
 \;\textrm{ and }\;
  \SAT[\phi] = \SAT[\psi_1] \diamond_1 R[0^{|\psi_1|}]
\end{equation}
where $\diamond_1$ is $\vee$ or $\wedge$ when $i = 0$ or $1$ respectively.

Fix two strings $a$ and $b$ such that $a \in R$ and $b \notin R$. If $|\psi_1| \geq |\phi|/2$ then $\SAT[\phi]$ is computable in polynomial time, and our reduction maps 
$\phi$ to either $a$ or $b$ depending on $\SAT[\phi]$ being $1$ or $0$ respectively. To put it differently, the right hand side of \eqref{second-reduction} will be substituted by
$R[a]$ or $R[b]$ respectively.

On the other hand, if $|\psi_1| < |\phi|/2$ then we repeat the same process for $\psi_1$.
We apply $f$ to $\psi_1$ to get 
\begin{equation}\label{third-reduction}
\SAT[\psi_1] = \SAT[\psi_2] \diamond_2 R[0^{|\psi_2|}] 
\end{equation}
By substituting this in \eqref{second-reduction}  we will have:
\begin{equation}\label{fourth-reduction}
\SAT[\phi] = (\SAT[\psi_2] \diamond_2 R[0^{|\psi_2|}]) \diamond_1  R[0^{|\psi_1|}] 
\end{equation}
Again, if $|\psi_2| \geq |\psi_1|/2$ then $\SAT[\psi_1]$ is computable in polynomial time, and its value can be substituted in \eqref{third-reduction} to get a reduction from $\SAT$ to $R$. On the other hand, if  $|\psi_2| < |\psi_1|/2$ then we use $f$ again to find $\psi_3$ such that:
\begin{equation}\label{fifth-reduction}
\SAT[\psi_2] = \SAT[\psi_3] \diamond_3 R[0^{|\psi_3|}] 
\end{equation}
By substituting this in \eqref{fourth-reduction} we will have:
\begin{equation}\label{sixth-reduction}
\SAT[\phi] = \big((\SAT[\psi_3] \diamond_3 R[0^{|\psi_3|}] ) \diamond_2 R[0^{|\psi_2|}]\big) \diamond_1  R[0^{|\psi_1|}] 
\end{equation}
We repeat this process up to $(\log  n)$ times where $n = |\phi|$. If there exists some $i \leq (\log  n)$ such that $ |\psi_{i+1}| \geq |\psi_i|/2$, then we can compute $\SAT[\psi_i]$ in polynomial time and substitute its value in the following equation:
\begin{equation}\label{k-reduction}
\SAT[\phi] = \big((\SAT[\psi_i] \diamond_k R[0^{|\psi_i|}] ) \diamond_{i-1} R[0^{|\psi_{i-1}|}]\big)... \diamond_1  R[0^{|\psi_1|}] 
\end{equation}
This gives us an $i$-$\tt$-reduction from $\SAT$ to $R$ for some $i < (\log  n)$.

On the other hand, if $|\psi_{i+1}| < |\psi_i|/2$ for every $i \leq (\log  n)$ then we will have:
\begin{equation}\label{final-reduction}
\SAT[\phi] = \big((\SAT[\psi_{(\log  n)}] \diamond_{(\log  n)} R[0^{|\psi_{(\log  n)}|}] ) \diamond_{(\log  n)-1} R[0^{|\psi_{(\log  n)-1}|}]\big)... \diamond_1  R[0^{|\psi_1|}] 
\end{equation}
It follows from the construction that the length of $\psi_i$'s is
halved on each step. Therefore $|\psi_{(\log n)}|$ must be constant in
$n$. As a result $\SAT[\psi_{(\log n)}]$ is computable in constant
time. If we compute the value of $\SAT[\psi_{(\log n)}]$, and
substitute it in \eqref{final-reduction} we will have a $(\log
n)$-$\tt$-reduction from $\SAT$ to $R$. In any case, we have shown
that if $\SAT$ is many-one reducible to $D$, we can use this reduction
to define a polynomial time computable $(\log n)$-$\tt$-reduction from
$\SAT$ to $R$. This means that $R$ is $(\log n)$-$\tt$-complete for
$\NP$.  Buhrman and van Melkebeek \cite{BuhvMe99} showed that complete
sets for $\NP$ under $\leq^\P_{n^{\alpha}-\tt}$-reductions have
$\ptwo$-measure $0$. Since this complete degree is closed under
$\leqpm$-reductions, it also has $\p$-measure 0 \cite{Lutz:WCE}.
Therefore the $(\log n)$-$\tt$-completeness of $R$ contradicts its
$\p$-randomness, which completes the proof.
\end{proof}

This next theorem is based on a result of Hitchcock and Pavan
\cite{Hitchcock:CRNPCS} that separated strong nondeterministic
completeness from Turing completeness for $\NP$. We separate nonuniform many-one completeness with polynomial advice from Turing completeness.

\begin{theorem}\label{th:polyadv-vs-T}
If the $\NP$-machine hypothesis holds, then there exists a
$\leqpmadv{\poly}$-complete set in $\NP$ that is not $\leqpT$-complete.
\end{theorem}
\begin{proof}
We follow the setup in \cite{Hitchcock:CRNPCS}.  Assume the
$\NP$-machine hypothesis holds. Then it can be shown there exists an
$\NP$-machine $M$ that accepts $0^*$ such that no $2^{n^3}$-time
bounded Turing machine can compute infinitely many of its
computations. Consider the following $\NP$ set:
\begin{equation}
A = \{\langle \phi , a \rangle \; | \; \phi \in \SAT \; \text{and $a$ is an accepting computation of } M(0^{|\phi|})  \}
\end{equation}
The mapping $ \phi \rightarrow \langle \phi, a \rangle$ where $a$ is
the first accepting computation of $M(0^{|\phi|})$ is a
$\leqpmadv{\poly}$-reduction from $\SAT$ to $A$. Note that $a$ only
depends on the length of $\phi$ and $|a|$ is polynomial in the
$|\phi|$. Therefore $A$ is $\leqpmadv{\poly}$-complete for $\NP$. It
is proved in \cite{Hitchcock:CRNPCS} that $A$ is not
$\leqpT$-complete.
\end{proof}

Because the measure hypothesis implies the $\NP$-machine hypothesis,
we have the following corollary.

\begin{corollary}
If $\mup(\NP) \neq 0$, then there exists a
$\leqpmadv{\poly}$-complete set in $\NP$ that is not $\leqpT$-complete.
\end{corollary}

\section{Separating Uniform Completeness from Nonuniform Completeness}\label{sec:uniform-nonuniform}
Buhrman et
al. \cite{BuhrmanHescottHomerTorenvliet:NonuniformReductions} showed
there is a $\leqp_{2-\tt}$-complete set for $\EXP$ that is not
$\leqpmadv{1}$-complete. We show the same for $\NP$-completeness under
a strong hypothesis on $\NP \cap\coNP$; in fact, the set is not even
complete with many-one reductions that use a fixed polynomial amount
of advice.
In the proof, we use the construction of a
$\leqp_{2-\tt}$-complete set that was previously used to separate
$\leqp_{2-\tt}$-completeness from $\leqp_{1-\tt}$-completeness
\cite{PavSel04} and $\leqp_{2-\tt}$-autoreducibility from
$\leqp_{1-\tt}$-autoredicibility \cite{Hitchcock:ANPC}.

\begin{theorem}
If $\mup (\NP \cap \coNP) \neq 0$ then for every $c \geq 1$, there exists a set $A \in \NP$
that is $\leqp_{2-\tt}$-complete but is not $\leqpmadv{n^c}$-complete.
\end{theorem}
\begin{proof}
We know that $\mup (\NP \cap \coNP) \neq 0$ implies $\muptwo (\NP \cap
\coNP) \neq 0$ \cite{Lutz:WCE}. Therefore we can assume there exists
$R \in \NP \cap \coNP$ that is $\ptwo$-random. We fix $c \geq 1$, and
define $A = 0(R \cap \SAT) \cup 1(\bar{R} \cap \SAT)$, where $\bar{R}$
is $R$'s complement.  It follows from closure properties of $\NP$ that
$A \in \NP$. We can define a polynomial-time computable
$2$-$\tt$-reduction from $\SAT$ to $A$ as follows: on input $x$ we
make two queries $0x$ and $1x$ from $A$, and we have $x \in \SAT
\leftrightarrow (0x \in A \vee 1x \in A)$. Therefore $A$ is
$\leqp_{2-\tt}$-complete in $\NP$. We will show that $A$ is not
$\leqpmadv{n^c}$-complete.  To get a contradiction, assume $A$ is
$\leqpmadv{n^c}$-complete in $\NP$. This implies that $R
\leqpmadv{n^c} A$ via functions $f \in \PF$ and $h : \N \rightarrow
\{0 , 1 \}^*$ where $(\forall n)\; |h(n)| = n^c$. In other words:
\begin{equation}\label{reduction1}
(\forall x) R[x] = A [ f\big( x , h(|x|) \big)] \;\;
 \text{where} \; |h(n)| = n^c
\end{equation}
For each length $n$ the advice $a_n$ has length $n^c$. As a result,
there are $2^{n^c}$ possibilities for $a_n$. For each length $n$ we
define $2^{n^c}$ martingales such that each martingale assumes one of
these possible strings is the actual advice for length $n$, and uses
\eqref{reduction1} to bet on $R$. We divide the capital into $2^{n^c}$
equal shares between these martingales. In the worst case, the
martingales that do not use the right advice lose their share of the
capital. We define these martingales such that the martingale that
uses the right advice multiplies its share by $2^{n^c+1}$. We will
also show that this happens for infinitely many lengths $n$, which
gives us a $\ptwo$-strategy to succeed on $R$. Note that based on the argument
above, we can only focus on the martingale that uses the right advice
for each length. To say it differently, in the rest of the proof we
assume that we know the right advice for each length, but the price
that we have to pay is to show that our martingale can multiply its
capital by $2^{n^c+1}$.

For each length $n$ we first compute $\SAT [z]$ for every string $z$ of length $n$. In particular, we are interested in the following set:
\begin{equation*}
A_n = \{ z \; | \; |z| = n \;\text{and} \; z \notin \SAT \}
\end{equation*}
If $|A_n| < n^{2c}$ we do not bet on any string of length $n$. It follows
from paddability of $\SAT$ that there must be infinitely many $n$'s
such that $|A_n| \geq n^{2c}$. Assume $n$ is a length where $|A_n| \geq
n^{2c}$, and let $a_n$ be the right advice for length $n$. For any string
$x$, let $v(0x) = v(1x) = x$. Consider the following set:
\begin{equation*}
C_n = \{ z \; |\; \; |z| = n, \; z \notin \SAT, \; \text{and} \; v \big(f(z , a_n) \big) > z \}
\end{equation*}
\begin{claim}\label{claim1}
There must be infinitely many $n$'s where $|A_n| \geq n^{2c}$ and $|C_n| \geq n^{2c} - n^c$. 
\end{claim}
\begin{proof}
Assume the claim does not hold. Then we have:
$(\forall^{\infty} n) \; |A_n| \geq n^{2c} \rightarrow |C_n| < n^{2c} - n^c$.
This means for almost every $n$ if $|A_n| \geq n^{2c}$  then there are $n^c+1$ strings of length $n$, $z_1, z_2,..., z_{n^c+1}$, satisfying the following property:
\begin{equation}\label{short1}
(\forall \; 1 \leq i \leq n^c+1) \; R[z_i] = A[f(z_i , a_n)] \;\wedge \; v(f(z_i , a_n)) \leq z_i
\end{equation}
It follows from the definition of $A$ that $A[y] = \tilde{R} \cap \SAT [v(y)]$ where $\tilde{R}$ is $R$ or $\bar{R}$ depending on whether $y$ starts with a $0$ or $1$ respectively. Therefore \eqref{short1} turns into:
\begin{equation}\label{short2}
(\forall \; 1 \leq i \leq n^c+1) \; R[z_i] = (\tilde{R} \cap \SAT)[v(f(z_i , a_n))] \;\wedge \; v(f(z_i , a_n)) \leq z_i
\end{equation}
We use \eqref{short2} to define a martingale that predicts $R[z_i]$ for every $1 \leq i \leq n^c+1$. Since we know  $R[z_i] = \tilde{R}[v(f(z_i , a_n))] \wedge \SAT[v(f(z_i , a_n))]$ our martingale computes $\tilde{R}[v(f(z_i , a_n))] \wedge \SAT[v(f(z_i , a_n))]$ and bets on $R[z_i]$ having the same value as $\tilde{R}[v(f(z_i , a_n))] \wedge \SAT[v(f(z_i , a_n))]$. Now we need to show why a polynomial time martingale has enough time to compute $\tilde{R}[v(f(z_i , a_n))] \wedge \SAT[v(f(z_i , a_n))]$.  Note that we know $ v(f(z_i , a_n)) \leq z_i$ so it is either the case that $ v(f(z_i , a_n)) < z_i$ or $ v(f(z_i , a_n)) = z_i$. In the first case, the martingale has access to $\tilde{R}[v(f(z_i , a_n))]$, and has enough time to compute $\SAT[v(f(z_i , a_n))]$. In the second case we know that $ \SAT[v(f(z_i , a_n))] = 0$ therefore $R[z_i] = 0$.  This implies that we can double the capital for each $z_i$. As a result, the capital can be multiplied by $2^{n^c+1}$. If this happens for infinitely many $n$'s we have a martingale that succeeds on $R$ which is a contradiction. This completes the proof of \rm{Claim} \ref{claim1}.
\end{proof}
The following claim states that when applying $f$ to elements of $C_n$ there cannot be many collisions. Define:
\begin{equation*}
D_n = \{ z \in C_n \; |\; (\exists \; y \in C_n) \; y < z \;\wedge \; f(y , a_n) = f(z , a_n) \}
\end{equation*}
\begin{claim}\label{collision}
There cannot be infinitely many $n$'s such that $|D_n| \geq n^c + 1$.
\end{claim}
\begin{proof}
To get a contradiction, assume there are infinitely many $n$'s such that $|D_n| \geq n^c+1$. Let $t_1, t_2,..., t_{n^c+1}$ be the first such strings. Then we have:
\begin{equation*}
(\forall \; 1 \leq i \leq n^c+1) \; (\exists \; r_i) \; r_i \in C_n \; \wedge \; r_i < t_i \; \wedge \; f(r_i , a_n) = f(t_i , a_n)
\end{equation*}
It follows that:
\begin{equation*}
(\forall \; 1 \leq i \leq n^c+1) \; (\exists \; r_i) \; r_i \in D_n \; \wedge \; r_i < t_i \; \wedge \; R[r_i] = R[t_i]
\end{equation*}
We can define a martingale that looks up the value of $R[r_i]$, and bets on $R[t_i]$ based on the equation above. This means that we can double the capital by betting on $R[t_i]$ for every $1 \leq i \leq n^c + 1$. As a result, the capital will be multiplied by $2^{n^c+1}$. If this happens for infinitely many $n$'s we will have a martingale that succeeds on $R$ which is a contradiction. This completes the proof of \rm{Claim} \ref{collision}.
\end{proof}
Assume $n$ is a length where $|C_n| \geq n^{2c} - n^c$. We have shown that there are infinitely many such $n$'s. We claim that for infinitely many of these $n$'s, since $R$ is $\ptwo$-random, there must be at least $(n^{2c} - n^c)/4 $ strings in $C_n$ that also belong to $R$.
\begin{claim}\label{randomness}
$(\forall ^{\infty} n) \; |C_n| \geq (n^{2c} - n^c) \;\rightarrow \;|C_n \cap R| \geq (n^{2c} - n^c)/4$.
\end{claim}
\begin{proof}
Assume this claim does not hold. Then we have:
\[(\exists^{\infty} n)\; |C_n| \geq n^{2c} - n^c \;\wedge \; |C_n \cap R| < (n^{2c} - n^c)/4\]
 We use this assumption to define a polynomial time martingale that succeeds on $R$.
We divide the original capital such that the martingale has $1/2n^2$ of the original capital for each length. Note that finding $n$'s where  $|C_n| \geq n^{2c} - n^c$ consists of computing $\SAT$ for every string of length $n$, and counting the number of negative answers, which can be done in at most $2^{3n}$ steps, followed by applying $f$ to these strings and comparing $v(f(z , a_n))$ and $z$, which can be done in time at most $2^{2n}$. This means a polynomial-time martingale has enough time to detect $C_n$'s where $|C_n| \geq n^{2c} - n^c$. After detecting these $C_n$'s we use a simple martingale that for every string $z$ in $C_n$ bets $2/3$ of the capital on $R[z] = 0$ and the rest on $R[z] = 1$. It is easy to verify that in the cases where $|C_n \cap R| < (n^{2c} - n^c)/4$ we win enough so the martingale succeeds on $R$. This completes the proof of \rm{Claim} \ref{randomness}.
\end{proof}
Let $n$ be a length where  $|C_n \cap R| \geq (n^{2c} - n^c)/4$, and consider the image of $C_n \cap R$ under $f(\cdot, a_n)$:
\begin{equation*}
I_n = \{f(z , a_n) \; | \; z \in C_n \cap R \}
\end{equation*}
It follows from \rm{Claim} \ref{collision} that $|I_n| \geq [(n^{2c} - n^c)/4] - n^c$. If we consider the image of $I_n$ under $v(\cdot)$ we have:
\begin{equation*}
V_n = \{v(f(z , a_n)) \; | \; z \in C_n \cap R \}
\end{equation*}
It is easy to see that $|V_n| \geq |I_n|/2$. Therefore for large enough $n$ we have
 $|V_n| \geq n^c+1$. Now if we use \eqref{reduction1} we have
 $R[z] = \tilde{R} \cap \SAT [v(f(z , a_n))].$
 We know that $z \in R$. This implies that $\tilde{R}[v(f(z , a_n))] = 1$. Therefore a martingale that bets on  $\tilde{R}[v(f(z , a_n))] = 1$ can double the capital each time. Since  $|V_n| \geq n^c+1$ this martingale multiplies the capital by $2^{n^c+1}$. As a result, we have a martingale that succeeds on $R$, which completes the proof.
\end{proof}

\section{Uniform Upper Bounds on Nonuniform Completeness}\label{sec:upperbound}
In this section, we consider whether nonuniformity can be removed in
$\NP$-completeness, at the expense of more queries.

Buhrman et
al. \cite{BuhrmanHescottHomerTorenvliet:NonuniformReductions} proved
that every $\leq^{\P/\log}_{\mathrm{T}}$-complete set for $\EXP$ is
also $\leqpT$-complete using a tableaux method. Hirahara
\cite{Hirahara15} proved a more general result that implies the same
for $\NP$.

\begin{theorem_cite}{Hirahara \cite{Hirahara15}}\label{th:upper-T}
Every $\leq^{\P/\log}_{\mathrm{T}}$-complete set in $\NP$ is $\leqpT$-complete.
\end{theorem_cite}

Valiant and Vazirani \cite{ValVaz86} proved that there exists a
randomized polynomial-time algorithm such that given any formula
$\phi$, outputs a list of formulas $l$ such that:

\begin{enumerate}
\item Every assignment that satisfies a formula in $l$ also satisfies $\phi$.
\item If $\phi$ is satisfiable, then with high probability at least one of the formulas in $l$ is uniquely satisfiable.
\end{enumerate}
Klivans and van Melkebeek \cite{KlivMe02} showed that Valiant-Vazirani lemma
can be derandomized if $\E^{\NP}$ contains a problem with exponential
$\NP$-oracle circuit complexity. This yields a deterministic
polynomial-time algorithm that given any $\phi$, outputs a list of
formulas $l$ such that:
\begin{enumerate}
\item Every assignment that satisfies a formula in $l$ also satisfies $\phi$.
\item If $\phi$ is satisfiable, then one of the formulas in $l$ is uniquely satisfiable.
\end{enumerate}

\begin{theorem}
  If $\E^{\NP}$ contains a problem with $\NP$-oracle circuit
  complexity $2^{\Omega(n)}$,
  then every $\leq^{\P/\log}_\tt$-complete
  set in $\NP$ is $\leqptt$-complete.
\end{theorem}
\begin{proof}
Let $A$ be an arbitrary $\leqpmadv{1}$-complete set in $\NP$. This
case includes most of the important details and makes describing the
proof simpler. We will extend to $\leqpttadv{\log}$ case later.
We will define a $\leqptt$-reduction from $\SAT$ to $A$. 

We define a padded version of $\SAT$ as follows:
 \begin{equation*}\label{sat-padd}
 \widehat{\SAT} = \{\phi 1 0^n \; | \; n \in \N \;\textrm{ and }\; \phi \in \SAT \}
 \end{equation*}
Then $\widehat{\SAT} \in \NP$, so $\widehat{\SAT} \leqpmadv{1} A$ via
some $f \in \PF$ and some $h:\N \rightarrow \{0,1\}$ where
$(\forall \phi)\; \widehat{\SAT}[\phi] = A[f \big( \phi , h(|\phi|) \big)].$

We will use $\widehat{\SAT}$ to pad formulas that have different
lengths, and make them of the same length. Fix an input formula $\phi$
over $n$ Boolean variables $x_1$,...,$x_n$, and let $m \in \N$ be
large enough such that all formulas $\phi \wedge x_1$, $\phi \wedge
\neg x_1$, $\phi \wedge x_1 \wedge x_2$, $\ldots$, and $\phi \wedge
\neg x_1 \wedge \neg x_2 \wedge \cdots \wedge \neg x_n$ can be padded
into formulas of length $m$. We denote the padded version of these
formulas by putting a bar on them. For example, the padded version of
$\phi \wedge x_1$ is denoted by $\overline{\phi \wedge x_1}$.

Before describing the rest of the algorithm, observe that the process
of reducing search to decision for a Boolean formula can be done using
independent queries in the case that the formula is uniquely
satisfiable. This is due to the fact that if a formula
$\psi(y_1,\ldots,y_m)$ is uniquely satisfiable, then for each $1 \le j
\le m$ exactly one of the formulas $\psi \wedge x_j$ and $\psi \wedge
\neg x_j$ is satisfiable. Therefore the unique satisfying assignment
can be found by making $m$ independent queries to $\SAT$, i.e.  $\psi
\wedge x_1, \ldots, \psi \wedge x_m$.

Using the hypothesis to derandomize the Valiant-Vazirani algorithm
\cite{KlivMe02}, we have a deterministic algorithm that on input
$\phi(x_1,\ldots,x_n)$ outputs a list containing polynomially many
formulas $\psi_1$,...,$\psi_m$ satisfying properties described above.
For each formula $\psi_j(y_1^j,...,y_{n_j}^j)$ consider $\psi_j \wedge
y_k^j$'s for every $1 \le k \le n_j$, and use padding in
$\widehat{\SAT}$ to turn these formulas into formulas of the same
length. We denote the padded version of $\psi_j \wedge y_k^j$ by
$\psi_j^k$ for simplicity. For each $\psi_j$ we make $n_j$ independent
queries to $A$: $q_1^j = f(\psi_j^1 , 0), \ldots, q_{n_j}^j =
f(\psi_j^{n_j} , 0)$. For each one of these queries if the answer is
positive we set the respective variable to $1$ and $0$ otherwise. We
repeat this process using $1$ as advice, and we will have $2m$
assignments. We argue that $\phi$ is satisfiable if and only if at
least one of these assignments satisfies it.  If $\phi$ is not
satisfiable then obviously none of these assignments will satisfy
it. On the other hand, if $\phi \in \SAT$ then 
at least one of the $\psi_j$'s must be uniquely satisfiable. In
this case the process described above will find this unique satisfying
assignment. Again, by the Valiant-Vazirani lemma we know that every
assignment that satisfies at least one of the $\psi_j$'s must also
satisfy $\phi$, which means one of the $2m$ assignments produced by
the algorithm above will satisfy $\phi$ in the case that $\phi$ is
satisfiable. It is evident from the algorithm that the queries are
independent. It is also easy to see that the reduction runs in
polynomial time in $|\phi|$ since we are applying a polynomial-time
computable function $f$ to arguments about the same length as $\phi$,
and we are doing this $2m$ times which is polynomial in
$|\phi|$. Therefore this algorithm defines a polynomial-time
truth-table reduction from $\SAT$ to $A$.

If the nonuniform reduction in the theorem above uses $k$ bits of
advice instead of considering two cases in the proof there are $2^k$
cases to be considered. If $k \in O(\log n)$ then this can be done in
polynomial time. Also note that the nonuniform reduction can be a
truth-table reduction instead of a many-one reduction, and the same
proof still works.
\end{proof}

The measure hypothesis on $\NP$ implies that $\E^\NP$ has high
$\NP$-oracle circuit complexity
\cite{AllStr94,Lutz:MLD2,Hitchcock:SSPP}. Therefore we have the following.

\begin{corollary}
  If $\mup(\NP) \neq 0$,
  then every $\leq^{\P/\log}_\tt$-complete
  set in $\NP$ is $\leqptt$-complete.  
\end{corollary}

\section{Hierarchy Theorems for Nonuniform Completeness}\label{sec:hierarchy}
We proved unconditionally that every 
$\leqpmadv{\log}$-complete
set in $\NP$ is $\leqpT$-complete. On the other hand, we showed that
under the $\NP$-machine hypothesis there exists a
$\leqpmadv{\poly}$-complete set in $\NP$ that is not
$\leqpT$-complete. This results in a separation of
$\leqpmadv{\poly}$-completeness from $\leqpmadv{\log}$-completeness
under the $\NP$-machine hypothesis.

\begin{theorem}\label{th:hierarchy}
If the $\NP$-machine hypothesis is true, then there exists a set in
$\NP$ that is $\leqpmadv{\poly}$-complete, 
but is not $\leqpTadv{\log}$-complete.
\end{theorem}
\begin{proof}
Assume the NP-machine hypothesis. From Theorem \ref{th:polyadv-vs-T},
we obtain a set that is $\leqpmadv{\poly}$-complete but not
$\leqpT$-complete. By Theorem \ref{th:upper-T}, this set cannot be $\leqpTadv{\log}$-complete.
\end{proof}

We have the following corollary because the measure hypothesis implies
the NP-machine hypothesis.

\begin{corollary}
If $\mup(\NP) \neq 0$, then there exists a set in
$\NP$ that is $\leqpmadv{\poly}$-complete, 
but is not $\leqpTadv{\log}$-complete.
\end{corollary}

We note that while Theorem \ref{th:hierarchy} is stated for many-one
vs. Turing, it applies to any reducibility in between.

\begin{corollary}
If the $\NP$-machine hypothesis is true, then for any reducibility
$\calR$ where $\leqpm$-reducibility implies $\calR$-reducility and
$\calR$-reducility implies $\leqpT$-reducibility, there is a set in
$\NP$ that is $\leq^{\P/\poly}_\calR$-complete, but is not
$\leq^{\P/\log}_\calR$-complete.
\end{corollary}

It is natural to ask if we can separate completeness notions above
$\Ppoly$ many-one. We observe that for this, we will need stronger
hypotheses than we have considered in this paper.
\begin{proposition} If there is a $\leq^{\Ppoly}_{\T}$-complete set that is
  not $\leqpmadv{\poly}$-complete in $\NP$, then $\NP \not\subseteq
  \Ppoly$.
\end{proposition}
\begin{proof}
If $\NP \subseteq \Ppoly$, then every set in $\NP$ is $\leqpmadv{\poly}$-complete.
\end{proof}

The measure hypothesis and the $\NP$-machine hypothesis are not known
to imply $\NP \not\subseteq \Ppoly$. If it is possible to separate
completeness notions above $\leqpmadv{\poly}$, it appears an
additional hypothesis at least as strong as $\NP \not\subseteq \Ppoly$
-- such as $\PH$ is infinite -- would be needed.

\bibliographystyle{plain}
\bibliography{main,rbm,dim}

\end{document}